\documentclass[12pt]{article}

\usepackage{amssymb,amsthm,amsmath}
\usepackage{algorithm}
\usepackage{algpseudocode}


\usepackage{subfig}
\usepackage[usenames,dvipsnames]{xcolor}
\usepackage[colorlinks,citecolor=blue,linkcolor=BrickRed]{hyperref}\usepackage[colorlinks,citecolor=blue,linkcolor=BrickRed]{hyperref}
\usepackage{fullpage}
\usepackage{graphicx}
\usepackage{xspace}
\usepackage{enumitem}
\usepackage{thmtools,thm-restate}
\usepackage{wrapfig}
\usepackage{comment}
\usepackage{cleveref}

\graphicspath{{figures/}}

\newtheorem{theorem}{Theorem}[section]
\newtheorem{corollary}[theorem]{Corollary}
\newtheorem{conjecture}[theorem]{Conjecture}
\newtheorem{lemma}[theorem]{Lemma}

\newtheorem{remark}[theorem]{Remark}

\newif\ifFULL
\FULLtrue

\newcommand{\IGNORE}[1]{}

\usepackage{tikz}
\usetikzlibrary{arrows}
\usetikzlibrary{arrows.meta}
\usetikzlibrary{shapes}
\usetikzlibrary{backgrounds}
\usetikzlibrary{positioning}
\usetikzlibrary{decorations.markings}
\usetikzlibrary{patterns}
\usetikzlibrary{calc}
\usetikzlibrary{fit}
\usetikzlibrary{snakes}
\tikzset{
    >=stealth',
    pil/.style={
           ->,
           thick,
           shorten <=2pt,
           shorten >=2pt,}
}
\tikzset{->-/.style={decoration={
  markings,
  mark=at position .5 with {\arrow{>}}},postaction={decorate}}}

\newcommand{\bm}{\mathbf}

\newcommand{\setR}{\mathbb{R}}

\newcommand{\p}{\ensuremath{\mathbb{P}}}

\newcommand{\E}{\mathbb{E}}

\newcommand{\eat}[1]{}
\newcommand{\hide}[1]{{\Large \color{red} Contents here are hidden! To reveal contents, remove this command.}}

\newcommand{\disc}{\ensuremath{\mathsf{disc}}}

\newcommand{\poly}{\ensuremath{\mathsf{poly}}}

\newcommand{\tr}{\mathsf{tr}}

\newcommand{\herdisc}{\mathsf{herdisc}}
\newcommand{\detlb}{\mathsf{detLB}}

\newcommand{\pvdisc}{\mathsf{pvdisc}}
\newcommand{\hpvdisc}{\mathsf{herpvdisc}}

\newcommand{\SDP}{\mathsf{SDP}}

\allowdisplaybreaks

        {\hspace*{\fill}$\Box$\par}

\title{A Tighter Relation Between Hereditary Discrepancy and Determinant Lower Bound}
\author{Haotian Jiang \thanks{University of Washington, Seattle, USA. \texttt{jhtdavid@cs.washington.edu.}} \\
\and Victor Reis \thanks{University of Washington, Seattle, USA. \texttt{voreis@cs.washington.edu. }}}
\date{\today}

\begin{document}

\begin{titlepage}
  \maketitle
  
  \begin{abstract}
In seminal work, Lov\'asz, Spencer, and Vesztergombi [European J. Combin., 1986] proved a lower bound for the hereditary discrepancy of a matrix $A \in \mathbb{R}^{m \times n}$ in terms of the maximum $|\det(B)|^{1/k}$ over all $k \times k$ submatrices $B$ of $A$. We show algorithmically that this determinant lower bound can be off by at most a factor of $O(\sqrt{\log (m) \cdot \log (n)})$, improving over the previous bound of $O(\log(mn) \cdot \sqrt{\log (n)})$ given by Matou\v{s}ek [Proc. of the AMS, 2013]. Our result immediately implies $\herdisc(\mathcal{F}_1 \cup \mathcal{F}_2) \leq O(\sqrt{\log (m) \cdot \log (n)}) \cdot \max(\herdisc(\mathcal{F}_1), \herdisc(\mathcal{F}_2))$, for any two set systems $\mathcal{F}_1, \mathcal{F}_2$ over $[n]$ satisfying $|\mathcal{F}_1 \cup \mathcal{F}_2| = m$. Our bounds are tight up to constants when $m = O(\mathrm{poly}(n))$ due to a construction of P\'alv\"olgyi [Discrete Comput. Geom., 2010] or the counterexample to Beck's three permutation conjecture by Newman, Neiman and Nikolov [FOCS, 2012].

  \end{abstract}
  \thispagestyle{empty}
\end{titlepage}

\section{Introduction}

Given a matrix $A \in \mathbb{R}^{m \times n}$, the {\em discrepancy} of $A$ is $\disc(A) := \min_{\bm{x} \in \{-1,+1\}^n} \|A \bm{x}\|_\infty$. 
The {\em hereditary discrepancy} of $A$ is defined as $\herdisc(A) := \max_{S \subseteq [n]} \disc(A_S)$, where $A_S$ denotes the restriction of the matrix $A$ to columns in $S$.  
For a set system $\mathcal{F}$, $\disc(\mathcal{F})$ and $\herdisc(\mathcal{F})$ are defined to be $\disc(A_\mathcal{F})$ and $\herdisc(A_\mathcal{F})$, where $A_\mathcal{F}$ is the incidence matrix of $\mathcal{F}$. 

In seminal work, Lov\'asz, Spencer, and Vesztergombi~\cite{lsv86} introduced a powerful tool, known as the {\em determinant lower bound}, for bounding hereditary discrepancy:
\begin{align*}
\detlb(A) := \max_{k \in \mathbb{N}} \max_{\substack{(S,T) \subseteq [m] \times [n] \\ |S| = |T| = k}} |\det(A_{S,T})|^{1/k} ,
\end{align*}
where $A_{S,T}$ denotes the restriction of $A$ to  rows in $S$ and columns in $T$. In particular, they showed that $\herdisc(A) \geq \frac{1}{2}  \detlb(A)$ for any matrix $A$. 
A reverse relation was established by Matous\v{e}k~\cite{m13}, who showed that $\herdisc(A) \leq O(\log(mn) \sqrt{\log (n)}) \cdot \detlb(A)$. 
However, Matous\v{e}k's bound does not match the largest known gap of $\Theta(\log (n))$ between $\herdisc(A)$ and $\detlb(A)$, given by a construction of P\'alv\"olgyi~\cite{p10} or the counter-example to Beck's three permutation conjecture~\cite{nnn12}.

Our main result is the following improvement over Matous\v{e}k's bound in~\cite{m13}. 

\begin{theorem} \label{thm:main}
Given a matrix $A \in \mathbb{R}^{m \times n}$, one can efficiently find $\bm{x} \in \{+1,-1\}^n$ such that $\| A \bm{x} \|_\infty \leq O(\sqrt{\log (m) \cdot \log (n)} \cdot \detlb(A))$. 
\end{theorem}

Restricting to an arbitrary subset of the columns of $A$, one immediately obtains the following:

\begin{corollary}
For any matrix $A \in \mathbb{R}^{m \times n}$, $\herdisc(A) \leq O(\sqrt{\log (m) \cdot \log (n)} \cdot \detlb(A))$. 
\end{corollary}

In light of the examples in \cite{p10,nnn12} where $\herdisc(A) \geq \Omega(\log n) \cdot \detlb(A)$, Theorem~\ref{thm:main} is tight up to constants whenever $m = \poly(n)$.
For the case where $m \gg \poly(n)$, one cannot hope to improve the $\sqrt{\log (m)}$ dependence on $m$ in Theorem~\ref{thm:main}.
In particular, the set system $\mathcal{F} = 2^{[n]}$ has $\herdisc(\mathcal{F}) = n$, $\detlb(\mathcal{F}) = \sqrt{n}$ and therefore $\herdisc(\mathcal{F}) \geq \sqrt{\log (m)} \cdot \detlb(\mathcal{F})$.  
It remains an open problem, however, whether one can improve the $\sqrt{\log n}$ factor in the later regime.

\paragraph{Hereditary discrepancy of union of set systems.} A question of V. S\'os (see~\cite{lsv86}) asks whether $\herdisc(\mathcal{F}_1 \cup \mathcal{F}_2)$ can be estimated in terms of $\herdisc(\mathcal{F}_1)$ and $\herdisc(\mathcal{F}_2)$, for any set systems $\mathcal{F}_1$ and $\mathcal{F}_2$ over $[n]$. 
This is, however, not possible without any dependence on $m = |\mathcal{F}_1 \cup \mathcal{F}_2|$ or $n$, as first shown by an example of Hoffman (Proposition 4.11 in \cite{m09}). This can also be seen from the examples in~\cite{p10,nnn12}. 
In \cite{kmv05}, it was shown that $\herdisc(\mathcal{F}_1 \cup \mathcal{F}_2) \leq O(\log (n)) \cdot \herdisc(\mathcal{F
}_1)$ when $\mathcal{F}_2$ contains a single set. 
For more general set systems, Matous\v{e}k~\cite{m13} proved that $\herdisc(\mathcal{F}) \leq O(\sqrt{t} \log(mn) \sqrt{\log(n)}) \cdot \max_{i \in [t]} (\herdisc(\mathcal{F}_i))$, where $\mathcal{F} = \mathcal{F}_1 \cup \cdots \cup \mathcal{F}_t$ and $m = |\mathcal{F}|$. 

Theorem~\ref{thm:main} together with Lemma 4 in~\cite{m13} immediately imply the following improvement of this result, whose proof is the same as in \cite{m13}. 
For $t=2$ and $m = \poly(n)$, this bound is tight up to constants. 

\begin{theorem} \label{thm:union}
Let $\mathcal{F}$ be a system of $m$ sets on $[n]$ such that $\mathcal{F} = \mathcal{F}_1 \cup \mathcal{F}_2 \cup \cdots \cup \mathcal{F}_t$. Then,
\begin{align*}
\herdisc(\mathcal{F}) \leq O\left( \sqrt{t \log (m) \log (n)} \right) \cdot \max_{i \in [t]} (\herdisc(\mathcal{F}_i)). 
\end{align*}
\end{theorem}

\paragraph{Approximating hereditary discrepancy.}
It was shown in~\cite{cnn11} that $\disc(A)$ cannot be approximated in polynomial time for an arbitrary matrix $A \in \{0,1\}^{m \times n}$. 
The more robust notion of hereditary discrepancy, however, can be approximated within a polylog factor.
The best-known result in this direction is a $O(\log(\min(m,n)) \cdot \sqrt{\log (m)})$-approximation to hereditary discrepancy via the $\gamma_2$-norm~\cite{mnt15}. 
When $ m = \poly(n)$, this approximation factor is $O(\log^{3/2} (n))$. 

Our result in Theorem~\ref{thm:main} suggests a potential approach of approximating hereditary discrepancy by approximating the determinant lower bound. 
There has been a recent line of work in approximating the maximum $k \times k$ subdeterminant for a given matrix $A$. For $k = \min(m,n)$, Nikolov~\cite{n15} gave a $2^{O(k)}$-approximation; for general values of $k$, Anari and Vuong~\cite{av20} showed a $k^{O(k)}$-approximation algorithm. 
If these results can be strengthened to a $2^{O(k)}$-approximation algorithm for general values of $k$, then together with Theorem~\ref{thm:main}, one would obtain the first $O(\log(n))$-approximation algorithm for hereditary discrepancy when $m = \poly(n)$.

\paragraph{Overview of proof of Theorem~\ref{thm:main}.}
We follow the approaches in \cite{b10} and \cite{m13}. 
The key notion to prove Theorem~\ref{thm:main} is that of {\em hereditary partial vector discrepancy}, which is defined as follows.  
Given a matrix $A \in \setR^{m \times n}$ with entries $a_{ij}$ for $i \in [m]$ and $j \in [n]$,
we consider the following SDP for a subset $S \subseteq [n]$ and a parameter $\lambda \geq 0$: 

\begin{equation*}\tag*{$\SDP(A, S, \lambda)$}
\begin{aligned}
 &\Big\|\sum_{j \in S} a_{ij} \bm{v}_j \Big\|_2^2 \le \lambda^2 \quad \forall i \in [m],\\
&\sum_{j=1}^n \|\bm{v}_j\|_2^2 \ge |S|/2, \\
 &\|\bm{v}_j\|_2^2 \leq 1 \quad \forall j \in S, \\
 &\|\bm{v}_j\|_2^2 = 0 \quad \forall j \in [n] \setminus S .
\end{aligned}
\end{equation*}

Define the {\em partial vector discrepancy} of $A$, denoted as $\pvdisc(A)$, to be the smallest value of $\lambda$ such that $\SDP(A, [n], \lambda)$ is feasible, and {\em hereditary partial vector discrepancy} $\hpvdisc(A)$ to be the smallest $\lambda$ such that $\SDP(A, S, \lambda)$ is feasible for any subset $S \subseteq [n]$. 

Using the above definition, we show in Lemma~\ref{lem:bansal} of \Cref{subsec:bansal} that the above SDP can be rounded efficiently to obtain a coloring with discrepancy at most $O(\sqrt{\log (m) \log (n)} \cdot \hpvdisc(A))$. We then prove in Lemma~\ref{lem:detlb_leq_hpvdisc} of \Cref{subsec:matousek} that $\hpvdisc(A) \le O(\detlb(A))$, from which Theorem~\ref{thm:main} immediately follows. 
We conjecture that $\hpvdisc(A)$ is the same as $\detlb(A)$ up to constants (Conjecture \ref{conj:detlb_hpvdisc}).

\paragraph{Notations and preliminaries.}
Given a matrix $A \in \setR^{m \times n}$, its rows will be denoted by $\bm{a}_1, \dots, \bm{a}_m \in \setR^n$. Define $A_{S,T}$ to be the matrix with rows restricted to some subset $S \subseteq [m]$ and columns restricted to some $T \subseteq [n]$, and $A_{S} := A_{[m],S}$.

\begin{theorem}[Freedman's Inequality, Theorem 1.6 in~\cite{f75}]
\label{thm:freedman}
Consider a real-valued martingale sequence $\{X_t\}_{t\geq 0}$ such that $X_0=0$, and $\E[X_{t+1}|\mathcal{F}_t]=0$ for all $t$, where $\{\mathcal{F}_t\}_{t\geq 0}$ is the filtration defined by the martingale. Assume that the sequence is uniformly bounded, i.e., $|X_t|\leq M$ almost surely for all $t$. Now define the predictable quadratic variation process of the martingale to be $W_t=\sum_{j=1}^t \E[X_j^2|\mathcal{F}_{j-1}]$ for all $t\geq 1$. Then for all $\ell \geq 0$ and $\sigma^2>0$ and any stopping time $\tau$, we have
\[
\p\Big[ \Big|\sum_{j=0}^\tau X_j \Big|\geq \ell \wedge W_\tau \leq \sigma^2 \text{for some stopping time } \tau \Big] \leq 2\exp\Big(- \frac{\ell^2/2}{\sigma^2+M \ell/3} \Big).
\]
\end{theorem}


\section{Proof of Theorem~\ref{thm:main}}

\subsection{The Algorithm} 
\label{subsec:bansal}

The main result of this subsection is the following lemma.

\begin{lemma} \label{lem:bansal}
Given a matrix $A \in \mathbb{R}^{m \times n}$, there exists a randomized algorithm that w.h.p. constructs a coloring $\bm{x} \in \{+1,-1\}^n$ such that $\| A \bm{x} \|_\infty \leq O(\sqrt{\log (m) \log (n)} \cdot \hpvdisc(A))$. This implies that 
$\herdisc(A) \le O(\sqrt{\log(m) \log (n)} \cdot \hpvdisc(A))$. 
\end{lemma}

The algorithm in Lemma~\ref{lem:bansal} is given in Algorithm~\ref{alg:rounding}. This algorithm is a variant of the random walk in \cite{b10}, using the SDP for hereditary partial vector discrepancy.

\begin{algorithm}[htp!]\caption{$\textsc{HerpvdiscRounding}(A)$}\label{alg:rounding}
\begin{algorithmic}[1]
\State $\lambda \leftarrow \hpvdisc(A)$ \Comment{The value of $\lambda$ can be approximated with a binary search}
\State $\bm{x}_0 \leftarrow \bm{0} \in \mathbb{R}^n$, $S_0 \leftarrow [n]$, $s \leftarrow 1/m^2 n^2$, $T \leftarrow 200 \log (n)/ s^2$
\For{$t = 1,2, \cdots, T$}
 	\State $\bm{v}_1, \cdots, \bm{v}_n \leftarrow \SDP(A, S_{t-1}, \lambda)$
 	\State Sample $\bm{r} \in \{-1,+1\}^n$ uniformly at random 
 	\For{$i \in [n]$}
 	 $x_t(i) \leftarrow x_{t-1}(i) + s \cdot \langle \bm{r}, \bm{v}_i \rangle$ 
 	\EndFor
 	\State $S_t \leftarrow S_{t-1}$
 	\For{$i \in [n] \setminus S_{t-1}$}
 		\If{$|x_t(i)| \geq 1 - 1/n$}
 			\State $S_t \leftarrow S_t \setminus \{i\}$
 		\EndIf
 	\EndFor
\EndFor
\State Round $\bm{x}_T$ to a vector $\bm{x} \in \{-1,+1\}^n$
\State \textbf{Return} $\bm{x}$
\end{algorithmic}
\end{algorithm}

Since Lemma~\ref{lem:bansal} is invariant under rescaling of the matrix $A$, we may assume without loss of generality that that $\max_{i,j}|a_{i,j}| = 1$. 
Given a coloring $\bm{x} \in [-1,1]^n$, we say an element $i \in [n]$ is alive if $|x(i)| < 1-1/n$. 
The following lemma from \cite{b10} states that the number of alive elements halves after $O(1/s^2)$ steps.

\begin{lemma}[Lemma 4.1 of \cite{b10}] \label{lem:half_colored}
Let $\bm{y} \in [-1,+1]^n$ be an arbitrary fractional coloring with at most $k$ alive variables. Let $\bm{z}$ be the fractional coloring obtained by running algorithm~\ref{alg:rounding} with $\bm{x}_0' = \bm{y}$ for $T' = 16/s^2$ steps. Then the probability that $\bm{z}$ has at least $k/2$  alive variables is at most $1/4$.
\end{lemma}

\begin{proof}[Proof of Lemma~\ref{lem:bansal}]
We first argue that after $T = 400 \log (n)/s^2$ steps, no element is alive with high probability. 
Divide the time horizon into epochs of size $16/s^2$. 
For each epoch, Lemma~\ref{lem:half_colored} states that regardless of the past, the number of alive elements decreases by at least half with probability at least $3/4$. 
It follows that no element is alive with high probability after $25 \log (n)$ epochs. 
Note that when no element is alive for the coloring $\bm{x}_T$, one can round it to a full coloring without changing the discrepancy of each set by more than $1$.

Next we prove that with high probability, the discrepancy of each row of $A$ is at most $O(\sqrt{\log (m) \log (n)}) \cdot \lambda$. 
We consider any $j \in [m]$, and denote $\disc_t(j) = \langle \bm{a}_j, \bm{x}_t \rangle$ the discrepancy of row $j$ at the end of time step $t \in [T]$. 
Note that $\E[\disc_t(j) - \disc_{t-1}(j) | \disc_{t-1}(j)] = 0$ and $\E[(\disc_t(j) - \disc_{t-1}(j))^2 | \disc_{t-1}(j)] \leq \lambda^2 s^2$. 
It follows from Freedman's inequality (Theorem~\ref{thm:freedman}) that
\begin{align*}
\p\left [|\disc_T(j)| \geq 10 \sqrt{\log (m) \log (n)} \cdot \lambda \right] \leq 1/m^2 .
\end{align*}
So by the union bound, the discrepancy of the obtained coloring is at most $O(\sqrt{\log(m) \log (n)} \cdot \hpvdisc(A))$ with high probability.
This completes the proof of Lemma~\ref{lem:bansal}. 
\end{proof}

\subsection{Bounding Partial Vector Discrepancy}
\label{subsec:matousek}

In this subsection, we prove the following lemma which upper bounds partial vector discrepancy in terms of the determinant lower bound. The proof can be seen as a simplification of Lemma 8 in~\cite{m13}, which gives a corresponding upper bound for {\em vector discrepancy} that is weaker by a factor of $\sqrt{\log n}$ due to a bucketing argument that is not needed here.

\begin{lemma} \label{lem:detlb_leq_hpvdisc}
For any $A \in \mathbb{R}^{m \times n}$, we have $\hpvdisc(A) \le O(\detlb(A)).$
\end{lemma}

\begin{proof}
Recall that $\pvdisc(A)^2$ is the optimal value of the SDP given by 
\begin{align*}
& \min \ t \\
& \Big\|\sum_{j=1}^n a_{ij} \bm{v}_j \Big\|_2^2 \le t \quad \forall i \in [m]\\
& \sum_{j=1}^n \|\bm{v}_j\|_2^2 \ge n/2 \\
& \|\bm{v}_j\|_2^2 \leq 1 \quad \forall j \in [n].
\end{align*} 

By denoting $X_{ij} := \langle \bm{v}_i, \bm{v}_j\rangle$, we may rewrite this SDP as follows:
\begin{align*}
& \min \ t \\
& \langle \bm{a}_i \bm{a}_i^\top, X \rangle \le t \quad \forall i \in [m]\\
& \langle I_n, X\rangle \ge n/2 \\
& \langle \bm{e}_j \bm{e}_j^\top, X\rangle \le 1 \quad \forall j \in [n] \\
& X \succeq 0,
\end{align*} 
where $\bm{e}_j$ denotes the vector with $1$ on the $j$-th coordinate and 0 elsewhere. The dual formulation of the above SDP is given by the following:
\begin{align*}
& \max \quad n \gamma - \sum_{j=1}^n z_j \\
 \quad &\sum_{i=1}^m w_i \bm{a}_i \bm{a}_i^\top + \sum_{j=1}^n z_j \bm{e}_j \bm{e}_j^\top \succeq 2 \gamma \cdot I_n \\
 &\sum_{i=1}^m w_i = 1 \\
 & \bm{w}, \bm{z} \ge 0.
\end{align*}

Denote $\lambda := \pvdisc(A)$. By Slater's condition, there exists a feasible dual solution $(\bm{w}, \bm{z}, \gamma)$ such that $\bm{w}, \bm{z} \geq 0$ and $n \gamma - \sum_{j=1}^n z_j = \lambda^2$. Indeed, the dual has a feasible interior point (for example, $w_i = 1/m, z_j = 1$ and $\gamma = 0$) and is bounded, since we may rewrite the first constraint as
\begin{align} \label{eq:constraint_1}
\sum_{i=1}^m w_i \bm{a}_i \bm{a}_i^\top \succeq \sum_{j=1}^n (2 \gamma - z_j) \cdot \bm{e}_j \bm{e}_j^\top,
\end{align}
 which implies 
\[
n \gamma - \sum_{j=1}^n z_j \le n \gamma - \frac{1}{2} \sum_{j=1}^n z_j \le \frac{1}{2} \tr\Big[\sum_{i=1}^m \bm{a}_i \bm{a}_i^\top\Big].
\]
Let $\tilde{A}$ be the matrix obtained from $A$ by multiplying the $i$-th row by $\sqrt{w_i}$ and $J \subseteq [n]$ be the set of columns for which $z_j < \frac{3}{2}\gamma$.
Note that $|J| \geq \frac{1}{3} n$, for otherwise $\sum_{j=1}^n z_j > \frac{2}{3} n \cdot \frac{3}{2} \gamma = n \gamma$.  
Since for each $j \in J$ we have $2\gamma - z_j \ge \frac{1}{2} \gamma$, for any vector $\bm{x} \in \mathbb{R}^J$ it follows by (\ref{eq:constraint_1}):
\begin{align*}
\bm{x}^\top \widetilde{A}_J^\top \widetilde{A}_J \bm{x} \geq \frac{1}{2}\gamma \cdot \| \bm{x} \|_2^2 \geq \frac{\lambda^2}{2n} \cdot \|\bm{x}\|_2^2 . 
\end{align*}
This implies that all eigenvalues of $\widetilde{A}_J^\top \widetilde{A}_J$ are at least $\lambda^2/2n$, so that $\det(\widetilde{A}_J^\top \widetilde{A}_J) \geq (\lambda^2/2n)^{|J|}$. In the other direction, the Cauchy-Binet formula also gives
\begin{align*}
\det(\widetilde{A}_J^\top \widetilde{A}_J) &= \sum_{\substack{I \subseteq [m] \\ |I| = |J|}} \det(\widetilde{A}_{I,J})^2 
= \sum_{\substack{I \subseteq [m] \\ |I| = |J|}} \det(A_{I,J})^2 \prod_{i\in I} w_i \\  
& \le \detlb(A)^{2|J|} \cdot  \sum_{\substack{I \subseteq [m] \\ |I| = |J|}} \prod_{i\in I} w_i 
 \le  \detlb(A)^{2|J|} \cdot \frac{1}{|J|!} \Big(\sum_{i=1}^m w_i\Big)^{|J|},
\end{align*}
where the last inequality follows as each term $ \prod_{i\in I} w_i$ appears $|J|!$ times in $\Big(\sum_{i=1}^m w_i\Big)^{|J|}$. Since $\sum_{i=1}^m w_i = 1$, we conclude 
\[
\detlb(A)^{2|J|} \cdot \frac{1}{|J|!} \ge \det(\widetilde{A}_J^\top \widetilde{A}_J) \geq (\lambda^2/2n)^{|J|},
\]
from which $\detlb(A) \ge \Omega(\lambda \cdot \sqrt{|J|/n}) = \Omega(\lambda) = \Omega(\pvdisc(A))$. Applying this result to all subsets $S \subseteq [n]$ of the columns of $A$ proves the lemma.
\end{proof} 

We conjecture that the above Lemma~\ref{lem:detlb_leq_hpvdisc} is tight up to constants. 

\begin{conjecture} \label{conj:detlb_hpvdisc}
For any matrix $A \in \mathbb{R}^{m \times n}$, we have $\detlb(A) = \Theta(\hpvdisc(A))$. 
\end{conjecture}



\section*{Acknowledgments}
We thank the anonymous reviewers of SOSA 2022 for insightful comments. 
We also thank Aleksandar Nikolov, Nikhil Bansal and Mehtaab Sawhney for helpful discussions.

\bibliographystyle{alpha}
\bibliography{bib.bib}

\begin{thebibliography}{KMV05}

\bibitem[AV20]{av20}
Nima Anari and Thuy-Duong Vuong.
\newblock An extension of pl{\"u}cker relations with applications to
  subdeterminant maximization.
\newblock In {\em Approximation, Randomization, and Combinatorial Optimization.
  Algorithms and Techniques (APPROX/RANDOM 2020)}. Schloss
  Dagstuhl-Leibniz-Zentrum f{\"u}r Informatik, 2020.

\bibitem[Ban10]{b10}
Nikhil Bansal.
\newblock Constructive algorithms for discrepancy minimization.
\newblock In {\em 2010 IEEE 51st Annual Symposium on Foundations of Computer
  Science}, pages 3--10. IEEE, 2010.

\bibitem[CNN11]{cnn11}
Moses Charikar, Alantha Newman, and Aleksandar Nikolov.
\newblock Tight hardness results for minimizing discrepancy.
\newblock In {\em Proceedings of the twenty-second annual ACM-SIAM symposium on
  Discrete Algorithms}, pages 1607--1614. SIAM, 2011.

\bibitem[Fre75]{f75}
David~A Freedman.
\newblock On tail probabilities for martingales.
\newblock {\em the Annals of Probability}, 3(1):100--118, 1975.

\bibitem[KMV05]{kmv05}
Jeong~Han Kim, Ji{\v{r}}{\'\i} Matou{\v{s}}ek, and Van~H Vu.
\newblock Discrepancy after adding a single set.
\newblock {\em Combinatorica}, 25(4):499, 2005.

\bibitem[LSV86]{lsv86}
L{\'a}szl{\'o} Lov{\'a}sz, Joel Spencer, and Katalin Vesztergombi.
\newblock Discrepancy of set-systems and matrices.
\newblock {\em European Journal of Combinatorics}, 7(2):151--160, 1986.

\bibitem[Mat09]{m09}
Jiri Matousek.
\newblock {\em Geometric discrepancy: An illustrated guide}, volume~18.
\newblock Springer Science \& Business Media, 2009.

\bibitem[Mat13]{m13}
Ji{\v{r}}{\'\i} Matou{\v{s}}ek.
\newblock The determinant bound for discrepancy is almost tight.
\newblock {\em Proceedings of the American Mathematical Society},
  141(2):451--460, 2013.

\bibitem[MNT14]{mnt15}
Ji{\v{r}}{{\'i}} Matou{\v{s}}ek, Aleksandar Nikolov, and Kunal Talwar.
\newblock Factorization norms and hereditary discrepancy.
\newblock {\em arXiv preprint arXiv:1408.1376}, 2014.

\bibitem[Nik15]{n15}
Aleksandar Nikolov.
\newblock Randomized rounding for the largest simplex problem.
\newblock In {\em Proceedings of the forty-seventh annual ACM symposium on
  Theory of computing}, pages 861--870, 2015.

\bibitem[NNN12]{nnn12}
Alantha Newman, Ofer Neiman, and Aleksandar Nikolov.
\newblock Beck's three permutations conjecture: A counterexample and some
  consequences.
\newblock In {\em 2012 IEEE 53rd Annual Symposium on Foundations of Computer
  Science}, pages 253--262. IEEE, 2012.

\bibitem[P{\'a}l10]{p10}
D{\"o}m{\"o}t{\"o}r P{\'a}lv{\"o}lgyi.
\newblock Indecomposable coverings with concave polygons.
\newblock {\em Discrete \& Computational Geometry}, 44(3):577--588, 2010.

\end{thebibliography}

\end{document}